\documentclass[12pt,leqno]{article}
\usepackage{epsfig}
\usepackage{graphics,graphicx}
 \usepackage{amsmath}
 \usepackage{amssymb}
\usepackage{amsthm}
\usepackage{hyperref}

\DeclareRobustCommand{\qed}{%
  \ifmmode 
  \else \leavevmode\unskip\penalty9999 \hbox{}\nobreak\hfill
  \fi
\quad\hbox{\qedsymbol}}
\newcommand{\mathbold}[1]{\mbox{\boldmath $#1$}}

\newcommand{\pr}{\mathbold P}

\newcommand{\ex}{\mathbold E}

\newcommand{\argmin}{\operatorname{argmin}}
\newtheorem{theorem}{Theorem}[section]

\begin{document}
\begin{center}
{\large Stepwise Choice of Covariates in High Dimensional Regression}

\quad\\
Laurie Davies\\
{\it Faculty of Mathematics\\
University of Duisburg-Essen, 45117 Essen, Federal Republic of
Germany\\
e-mail:laurie.davies@uni-due.de}
\end{center}
\quad\\






\begin{abstract}
Given data ${\mathbold y}(n)$ and $q(n)$ covariates ${\mathbold x}(n)$
one problem in linear regression 
is to decide which if any of the covariates to include. There are many
articles on this problem but all are based on a stochastic model for the
data.  This paper gives what seems to be a new approach  which does
not require any form of model.  It is conceptually and algorithmically
simple. Instead of testing whether a regression parameter $\beta_{\nu}$
is zero it asks to what extent the corresponding covariate ${\mathbold
  x}_{\nu}$ is better than Gaussian noise as measured by the
probability of a greater reduction in the sum of squared residuals. An
exact expression for this probability is available and consistency
results can be proved under appropriate assumptions. The idea can be
extended to non-linear and robust regression.
\end{abstract}



\section{Introduction} \label{sec:intro}
Most if not all approaches for choosing covariates in high
dimensional linear regression are based on the model 
\begin{equation} \label{equ:standard_model}
{\mathbold Y}(n)={\mathbold X}(n){\mathbold \beta}(n) +{\mathbold \varepsilon}(n)
\end{equation}
where just $k(n)$ of the $\beta_j(n), j=1,\ldots,q(n)$ are non-zero or
large, the remainder being zero or small. The following approach
differs in that it is not based on the model
(\ref{equ:standard_model}) or indeed any model. Whether a covariate is
included or not depends only on the degree to which it is better than
standard Gaussian white noise.

More precisely, suppose that at one stage of the stepwise procedure
a subset of the covariates of size $k_0\le n-2$ has been
included in the regression. Denote the indices of this subset by
${\mathcal S}_0$ and the mean sum of squared residuals by
$ss_0$. Now include the covariate ${\mathbold
  x}_{\nu}=(x_{i\nu})_{i=1}^n$ with $\nu \notin {\mathcal
  S}_0$ and denote the mean sum of squared residuals based on ${\mathcal
  S}_0\cup \{\nu\}$ by $ss_{\nu}$. Including the best of the
covariates not in ${\mathcal S}_0$ leads to a minimum mean sum of squared
residuals 
\[ ss_{01} =\min_{\nu\notin {\mathcal S}_0} ss_{\nu}.\]
The covariates not in ${\mathcal S}_0$ are now replaced in their
entirety by standard Gaussian white noise. Including the random
covariate corresponding to ${\mathbold x}_{\nu}$ leads to a random
mean sum of squared residuals $SS_{\nu}$ and including the best of the
random covariates to a minimum mean sum of squared residuals
\[ SS_{01}=\min_{\nu\notin {\mathcal S}_0} SS_{\nu}.\]
The probability that the random covariates are better than the actual
ones given ${\mathcal S}_0$ is
\begin{eqnarray*}
\pr(SS_{01} <ss_{01})&=& 1-\pr(SS_{01}  \ge ss_{01})=1-\pr(\min_{j  \notin
  {\mathcal S}_0} SS_j \ge ss_{01})\\
&=&1-\prod_{j \notin {\mathcal S}_0}\pr(SS_j\ge ss_{01})\\
\end{eqnarray*}
It can be shown that
\begin{equation} \label{equ:beta_p}
SS_j\stackrel{D}{=} ss_0(1-B_{1/2,(n-\nu_0-1)/2})
\end{equation}
where $B_{a,b}$ denotes a beta random variable with parameters $a$ and
$b$. This was a
personal communication from Lutz D\"umbgen, his proof is given in the
appendix. It replaces the approximation based on the chi-squared
distribution which was used in earlier versions of this paper. Thus
\[\pr(SS_j\ge ss_{01})=\text{pbeta}(1-ss_{01}/ss_0,1/2,(n-\nu_0-1)/2)\]
so that finally
\begin{equation} \label{equ:pval_1_exact}
 \pr(SS_{01}\le ss_{01})=
 1-\text{pbeta}(1-ss_{01}/ss_0,1/2,(n-\nu_0-1)/2)^{q(n)-\nu_0}.
\end{equation}
where $\text{pbeta}(\cdot,a,b)$ denotes the distribution function of $B_{a,b}$
This is the $p$-value for the inclusion of the next covariate.  
It is worth noting that (\ref{equ:pval_1_exact}) takes into account
the number  $\nu_0$ of covariates already active  whose influence is
small and the total number of covariates $q(n)-\nu_0$ not yet included
whose influence is large for large $q(n)$. 
 
If $\nu$ covariates have been included with $p$-values
$p_1,\ldots,p_{\nu}$ then the probability that each and every one  of
them is better than Gaussian noise is
\begin{equation} \label{equ:pval_2}
\prod_{j=1}^{\nu}(1-p_j) \,.
\end{equation}
This is because at each stage independent Gaussian covariates are
used.

One proposal is simply to calculate the $p$-values
(\ref{equ:pval_1_exact}) until the first one that exceeds a given
$\alpha$, say $\alpha=0.01$, and then include all previous
covariates. This leads to the stopping rule
\begin{equation} \label{equ:stop_rule_1}
ss_{01}> ss_0\left(1-\text{qbeta}((1-\alpha)^{1/(q(n)-\nu_0)},1/2,(n-\nu_0-1)/2)\right)
\end{equation}
where $\text{qbeta}(\cdot,a,b)$ is the quantile function of the beta
distribution with parameters $a$ and $b$.  The asymptotic version for
large $n$ and $q(n)$ is
\begin{eqnarray}
\lefteqn{ss_{01}>}\label{equ:asymp_p}\\
&&ss_0\Big(1-(2\log
q(n)-\log\log(q(n))-2\log(-\log(1-\alpha)))/n\Big).\nonumber
\end{eqnarray}
A proof is given in the appendix.

The calculation of the $p$-values (\ref{equ:pval_1_exact}) does not require
the choice of a regularization parameter. There is no need for
cross-validation or indeed any form of simulation. Furthermore, as the
procedure is not based on the linear model (\ref{equ:standard_model})
it does not require an estimate for the error variance $\sigma^2$. The
method is invariant with respect to affine changes of 
units and to permutations of the covariates. There are no problems of
multiple testing as this is covered by (\ref{equ:pval_2}). The method
can be extended to robust regression, to non-linear approximations of
the form $g({\mathbold   x}^T{\mathbold \beta})$ if $g$ has a Taylor
expansion and to the Kullback-Leibler discrepancy where this is
appropriate. In these extensions there is no exact form for the
$p$-values corresponding to (\ref{equ:pval_1_exact}) but there exist
simple approximations based on the chi-squared distribution.


This paper is based on \cite{DAV16a,DAV16b}. 
 The
paper \cite{DAV16a} goes beyond the choice of covariates and considers
non-significance regions in place of confidence regions for the values
of the parameter ${\mathbold \beta}$. This will not be considered here.

Stepwise regression is treated in Section~2; linear least squares
regression in Section~2.1, $M$-regression in Section~2.2 and
non-linear approximation in Section~2.3. A consistency result for
least squares linear regression is proved in Section~3.1 and the false
discovery rate is considered in Section~3.2. Results of simulations 
following \cite{DEGBUERITDEZ14} and \cite{JIAROH15} are given in
Section~4.1 and results for in Section~4.2. The proofs of (\ref{equ:beta_p})
and (\ref{equ:asymp_p}) are given in an appendix.

\section{Stepwise regression}
\subsection{Least squares regression} \label{sec:2.1}

As an example we take the leukemia data (\cite{GOLETAL99}
\begin{center}
http://www-genome.wi.mit.edu/cancer/\\
\end{center}
which was analysed in \cite{DETBUH03}. The number of patients is $n=72$
with $q(n)=3571$ covariates. The dependent variable takes on only the
values 0 and 1 depending on whether the patient suffers from acute
lymphoblastic leukemia or acute myeloid leukemia. The first five
genes in order of inclusion with their associated $p$-values as
defined by (\ref{equ:pval_1_exact}) are as follows:
\begin{equation} \label{equ:leukemia_lsq_proc1}
\begin{tabular}{cccccc}
gene number&1182&1219& 2888& 1946&2102\\
$p$-value&0.0000&8.58e-4&3.58e-3&2.54e-1&1.48e-1\\
\end{tabular}
\end{equation}

 According to this relevant genes are 1182, 1219 and 2888. Given these
 three genes the remaining 3568 are individually no better than random noise with
 respect to the two forms of leukemia. In particular gene 1946 is no
 better than random Gaussian noise.

 A linear regression based on the genes 1182, 1219, 2888 and
 1946 results in the $p$-values 0.0000, 9.84e-8, 5.74e-7 and
 8.20e-5 respectively. The difference between the $p$-value  8.20e-5 and the
 $p$-value 0.0254 for gene 1946 is that the latter takes into account
 that the gene 1946 is the best of 3568 genes and the former does not.

 The time required was 0.1 seconds using a Fortran 77
 programme linked to R. The source code for beta distribution function
 was a double precision version of that given in  \cite{PRTEVEFL03}. 

Interest in the above data centred on classifying patients based on
their gene expression data. The $p$-values suggest that for this it is
only necessary to base the classification on the genes 1182, 1219 and
2888. A simple linear regression based on genes 1182,1219 and 2888
results in one misclassification.

If the fourth covariate 1946 is also included when classifying the patients
there are no misclassifications. However if all the covariates except
1182, 1219 and 2888 are replaced by Gaussian white noise then in about
7\% of the cases including the fourth covariate results in no
misclassifications. Including the fourth and fifth covariates results
in no misclassifications in about 60\% of the cases. This is in spite
of the fact that  additional covariates are no more than Gaussian noise.

The above does not imply that 1182, 1219 and 2888 are the only relevant
genes. To see this remove these genes and repeat the analysis. This
leads to 
\begin{equation} \label{equ:leukemia_lsq-1182}
\begin{tabular}{ccccc}
gene number&1652&979&657&2260\\
$p$-value&0.0000&9.36e-5&2.75e-2&2.22e-2.\\
\end{tabular}
\end{equation}
The genes 1652, 979, 657 and 2260 can now be removed and the process
continued. If this is done it  results in 281 genes which are possibly
relevant.

\subsection{$M$-regression}
The method can  be applied to $L_1$ regression but with the
disadvantage that there does not exist a simple expression
corresponding to (\ref{equ:pval_1_exact}). If there is a particular
interest in $L_1$ regression simulations will be required. If however
$L_1$ regression is only used as a protection against outlying
$y$-values this can be provided by $M$-regression for which an
approximate version of (\ref{equ:pval_1_exact}) involving the
chi-squared distribution is available.

Let $\rho$ by a symmetric, positive and twice differentiable convex
function with $\rho(0)=0$.  The default function will be the Huber's
$\rho$-function with a tuning constant $c$ (\cite{HUBRON09}, page 69)
defined by
\begin{equation}
\rho_{c}(u)=\left\{\begin{array}{ll}
\frac{u^2}{2}, &\vert u\vert \le c,\\
c\vert u\vert -\frac{c^2}{2},&\vert u \vert > c.\\
\end{array}
\right.
\end{equation}
The default value of $c$ will be $c=1$.

For a given subset ${\mathcal S}_0$ of size $\nu_0$ the sum of squared
residuals is replaced by
\begin{equation}  \label{equ:min_rho}
s_0(\rho,\sigma)=\min_{{\mathbold
    \beta}({\mathcal S}_0)}\,\frac{1}{n}\sum_{i=1}^n 
\rho\left(\frac{y_i-\sum_{j \in{\mathcal S}_0}x_{ij}\beta_j({\mathcal
      S}_0)}{\sigma}\right).
\end{equation}
which can be calculated using the algorithm described in {\bf 7.8.2} of
\cite{HUBRON09}. The minimizing $\beta_j({\mathcal S}_0)$ will de
denoted by ${\tilde \beta}_j({\mathcal S}_0)$. A proposal for the
choice of $\sigma$  is given below.

For some $\nu \notin {\mathcal S}_0$  put 
\begin{equation}  \label{equ:min_rho_j}
s_{\nu}(\rho,\sigma)=\min_{{\mathbold
    \beta}({\mathcal S}_0\cup \{\nu\})}\,\frac{1}{n}\sum_{i=1}^n 
\rho\left(\frac{y_i-\sum_{j \in{\mathcal S}_0\cup\{\nu\}}x_{ij}\beta_j({\mathcal
      S}_0\cup\{\nu\})}{\sigma}\right).
\end{equation}
and
\begin{equation}
s_{01}(\rho,\sigma) =\min_{\nu \notin {\mathcal S}_0}s_{\nu}(\rho,\sigma).
\end{equation}

 Replace all the covariates not in ${\mathcal S}_0$ by standard
 Gaussian white noise, include the $\nu$th random covariate
 denoted by $Z_i$ and put
\begin{equation} 
S_{\nu}(\rho,\sigma)=\min_{{\mathbold \beta}({\mathcal S}_0),b}\,\frac{1}{n}\sum_{i=1}^n  
\rho\left(\frac{y_i-\sum_{j \in{\mathcal S}_0}x_{ij}\beta_j({\mathcal 
      S}_0)-bZ_i}{\sigma}\right).
\end{equation}
A Taylor expansion gives
\begin{eqnarray}
S_{\nu}(\rho,\sigma)&\approx&\frac{1}{2}\frac{\left(\sum_{i=1}^n\rho^{(1)}\left(\frac{r_i}{\sigma}\right)Z_i\right)^2}{\sum_{i=1}^n\rho^{(2)}\left(\frac{r_i}{\sigma}\right)Z_i^2}\nonumber\\
&\approx&s_0(\rho,\sigma)-\frac{1}{2}\frac{\left(\sum_{i=1}^n\rho^{(1)}\left(\frac{r_i}{\sigma}\right)\right)^2}{\sum_{i=1}^n\rho^{(2)}\left(\frac{r_i}{\sigma}\right)}\chi^2_1
\end{eqnarray}
with $r_i=y_i-\sum_{j \in{\mathcal S}_0}x_{ij}{\tilde \beta}_j({\mathcal S}_0)$.
This leads to the asymptotic $p$-value  
\begin{equation} \label{equ:pval_m}
1-\text{pchisq}\left(
  \frac{2s_0(\rho^{(2)},\sigma)}{s_0(\rho^{(1)},\sigma)}
  (s_0(\rho,\sigma)-s_{01}(\rho,\sigma))\right)^{q(n)-\nu_0}.   
\end{equation}
corresponding to the exact $p$-value (\ref{equ:pval_1_exact}) for linear
regression. Here 
\[ s_0(\rho^{(1)},\sigma)= \frac{1}{n}\sum_{i=1}^n 
\rho^{(1)}\left(\frac{r_i}{\sigma}\right)^2 , \quad
s_0(\rho^{(2)},\sigma)= \sum_{i=1}^n  
\rho^{(2)}\left(\frac{r_i}{\sigma}\right).\]

It remains to specify the choice of scale $\sigma$. The initial value
of $\sigma$  is the median absolute deviation of ${\mathbold y}(n)$
multiplied by the Fisher consistency factor 1.4826. If the ${\mathbold
  y}(n)$ have a large atom more care is needed. Let $n_{\Delta}$ be
the size of the   largest atom. Instead of the median of the absolute
deviations from   the median we now take the $0.5(n+n_{\Delta})/n$
quantile (see Chapter 4 of \cite{DAV14}). Let $\sigma_0$ denote the scale
at some point of the procedure with $\nu_0$ covariates already
included. After the next covariate has been 
included the new scale $\sigma_1$ is taken to be
\begin{equation} \label{equ:sig_rob_reg}
\sigma_1^2=\frac{1}{(n-\nu_0-1)c_f}\sum_{i=1}^n \psi(r_1(i)/\sigma_0)^2
\end{equation}
where the $r_1(i)$ are the residuals based on the $\nu_0+1$ covariates
and $c_f$ is the Fisher consistency factor given by
\[ c_f=\ex(\psi(Z)^2)\]
where $Z$ is $N(0,1)$ (see \cite{HUBRON09}). Other choices are also
possible.

\subsection{Non-linear approximation} \label{sec:non_lin}
The dependent variable ${\mathbold y}(n)$ is now approximated by $g({\mathbold
  x}({\mathcal S})^T{\mathbold \beta}({\mathcal S}))$ for some smooth
function $g$. Consider a subset ${\mathcal S}_0$, write
\begin{equation}
ss_0= \min_{{\mathbold \beta}({\mathcal S}_0)}\,\frac{1}{n}\sum_{i=1}^n (y_i-g({\mathbold
  x_i}({\mathcal S}_0)^T{\mathbold \beta}({\mathcal S}_0)))^2.
\end{equation}
and denote the minimizing ${\mathbold \beta}({\mathcal S}_0)$ by
${\tilde {\mathbold \beta}}({\mathcal S}_0)$. 
Now include one additional covariate ${\mathbold x}_{\nu}$ with $\nu \notin {\mathcal
  S}_0$  to give ${\mathcal S}_1={\mathcal S}_0\cup\{\nu\}$, denote
the mean sum of squared residuals by $ss_{\nu}$ and the minimum over
all possible choice of $\nu \notin {\mathcal S}_0$ by $ss_1$. As before all
covariates not in ${\mathcal S}_0$ are replaced by standard Gaussian
white noise. Include the $\nu$th random covariate denoted by $Z_i$, put
\[SS_{\nu}=\min_{\beta({\mathcal S}_0),b}\frac{1}{n}\sum_{i=1}^n (y_i-g({\mathbold
  x_i}({\mathcal S}_0)^T{\mathbold \beta}({\mathcal S}_0)+bZ_i))^2
\]
and denote the minimum over all possible choice of $\nu \notin
{\mathcal S}_0$ by $SS_1$.

Arguing as in the last section for robust regression results in 
\begin{equation} \label{equ:lsq_non_lin_1}
SS_1\approx ss_0- \frac{\sum_{i=1}^n r_i({\mathcal
     S}_0)^2g^{(1)}({\mathbold x}_i({\mathcal S}_0)^T{\tilde {\mathbold
     \beta}}({\mathcal S}_0))^2}{\sum_{i=1}^ng^{(1)}({\mathbold
     x}_i({\mathcal S}_0)^T{\tilde {\mathbold 
     \beta}}({\mathcal S}_0))^2}\chi^2_1
\end{equation}
where 
\begin{equation}\label{equ:lsq_non_lin_2}
r_i({\mathcal S}_0)=y_i-g({\mathbold
  x_i}({\mathcal S}_0)^T{\tilde {\mathbold \beta}}({\mathcal S}_0)).
\end{equation}
The asymptotic $p$-value corresponding to the asymptotic $p$-value
(\ref{equ:pval_m}) for $M$-regression is 
\begin{equation} \label{equ:pval_non_lin}
1-\text{pchisq}\left(\frac{(ss_0-ss_1)\sum_{i=1}^ng^{(1)} ({\mathbold
      x}_i({\mathcal S}_0)^T{\tilde {\mathbold \beta}}({\mathcal
      S}_0))^2} {\sum_{i=1}^n r_i({\mathcal
     S}_0)^2g^{(1)}({\mathbold x}_i({\mathcal S}_0)^T{\tilde {\mathbold
     \beta}}({\mathcal S}_0))^2},1\right)^{q(n)-\nu_0}.
\end{equation}

If $g$ is the logistic function $g(u)=\exp(u)/(1+\exp(u))$ 
\begin{eqnarray}
\lefteqn{\hspace{-2cm}\frac{\sum_{i=1}^n r_i({\mathcal
     S}_0)^2g^{(1)}({\mathbold x}_i({\mathcal S}_0)^T{\mathbold
     \beta}({\mathcal S}_0))^2}{\sum_{i=1}^ng^{(1)}({\mathbold
     x}_i({\mathcal S}_0)^T{\mathbold  \beta}({\mathcal
     S}_0))^2}}\nonumber\\
\hspace{2cm}&=&\frac{\sum_{i=1}^n
   (y_i-p_i(0))^2p_i(0)^2(1-p_i(0))^2}{\sum_{i=1}^np_i(0)^2(1-p_i(0))^2} 
\label{equ:logistic}
\end{eqnarray}
where 
\[p_i(0)=\frac{\exp({\mathbold x}_i({\mathcal S}_0)^T{\mathbold
     \beta}({\mathcal S}_0))}{1+\exp({\mathbold x}_i({\mathcal S}_0)^T{\mathbold
     \beta}({\mathcal S}_0))}.\]
The opportunity is now taken to correct an error in
\cite{DAV14}. The term  
\[\frac{\sum_{i=1}^np_i^3(1-p_i)^3}{\sum_{i=1}^np_i^2(1-p_i)^2}\]
occurs repeatedly in  Chapter~11.6.1.2 and should be replaced by
\[\frac{\sum_{i=1}^n(y_i-p_i)^2p_i^2(1-p_i)^2}{\sum_{i=1}^np_i^2(1-p_i)^2}\]
agreeing with (\ref{equ:logistic}). 

Robust non-linear regression can be treated in the same manner but the
expressions become unwieldy.

\subsection{Kullback-Leibler and logistic regression} \label{sec:kul_leib}
In some data sets, for example the leukemia data of Section~2.1, the
dependent variable ${\mathbold y}(n)$ takes on only the values zero and
one. For such data least squares combined with the logistic model can
cause problems: it can happen that for some $i$ the estimated
probability is $p_i\ge 1-10^{-10}$ although $y_i=0$. An example is
provided by the colon cancer data (\cite{ALOETAL99},
http://microarray.princeton. edu/oncology/.). The sample size is
$n=62$ and there are 2000 covariates. The logistic model based on
least squares with a cut-off $p$-value of 0.01 results in two
covariates. If these two covariates are used to classify the cancer
there are three errors. In two cases the probability based on the
logistic model is one whereas the dependent variable has the value
zero. In the third case the values are zero and one respectively.

The problem can be avoided by using the Kullback-Leibler discrepancy 
\begin{equation} \label{equ:kul_leib_1}
kl({\mathbold y}(n),{\mathbold
  x}(n),\mathbold{\beta})=-\sum_{i=1}^n(y_i\log p({\mathbold
  x}_i,{\mathbold \beta})+(1-y_i)\log(1- p({\mathbold
  x}_i,{\mathbold \beta})))
\end{equation}
where
\[p({\mathbold x}_i,{\mathbold \beta})=\frac{\exp({\mathbold
    x}_i^t{\mathbold \beta})}{1+\exp({\mathbold
    x}_i^t{\mathbold \beta})}.
\]

The arguments of the previous two sections lead to the  asymptotic $p$-values
\begin{equation} \label{equ:asym_p_kl}
1-\text{pchisq}\left(\frac{2\sum_{i=1}^n p_i(0)(1-p_i(0))}{\sum_{i=1}^n(y_i-p_i(0))^2}(kl_0-kl_1)\right)^{q(n)-p(0)}.
\end{equation}
where $kl_0$ is the minimum  Kullback-Leibler discrepancy based on a
subset ${\mathcal S}_0$  and  $kl_1$ the minimum value through the
inclusion of one additional covariate. The $p_i(0)$ are the values of
$p({\mathbold x}_i,{\mathbold \beta})$ giving the minimum $kl_0$.


Repeating the least squares analysis for the colon data but now
using the Kullback-Leibler divergence results in the single gene
number 377. The number of misclassifications based on this gene is
nine. The second gene 356 has a $p$-value of 0.053. If this is included
the number of misclassifications is reduced to two but the status of
this second gene is not clear.

\section{Consistency and false discovery rate}
\subsection{Consistency: $q(n)>k(n)$}
To prove a consistency result a model with error term is necessary. It
is defined as follows. There are $p$ covariates ${\mathbold
  x}_1,\ldots,{\mathbold x}_p$ and the response ${\mathbold Y}$ is
given by the first $k$
\begin{equation}
{\mathbold Y}=\sum_{j=1}^k {\mathbold x}_j\beta_j+{\mathbold \varepsilon}
\end{equation}
where the  ${\mathbold \varepsilon}$ are i.i.d.
Gaussian random variables with zero mean and variance $\sigma^2$.  Given a
sample of size $n$  both $k=k(n) =o(n)$ and $q=q(n)$ will be allowed to
depend on $n$. Without loss of generality the covariates ${\mathbold
  x}_j(n)$ will be standardized to have $L_2$ norm $\sqrt{n}$, that is $\Vert
{\mathbold x}_j(n)\Vert_2^2=n$. 

A given subset of $\{1,\ldots,k(n)\}$
will be denoted by ${\mathcal  S}_0(n)$. For $\nu \notin
{\mathcal S}_0(n)$  the subset ${\mathcal S}_0(n)\cup\{\nu\}$ will be
denoted by ${\mathcal S}_{01}(n)$, the complements of ${\mathcal
  S}_0(n)$ and ${\mathcal S}_{01}(n)$ with respect to
$\{1,\ldots,k(n)\}$ by ${\mathcal S}_{12}(n)$ and ${\mathcal S}_2(n)$
respectively. The projection onto the linear space spanned by the
${\mathbold x}_j(n),\, j \in {\mathcal S}_0(n),$ will be denoted by
$P_0$ with the corresponding notation for ${\mathcal S}_{01}(n)$,
${\mathcal S}_{12}(n)$ and ${\mathcal S}_2(n)$. Finally put 
\[{\mathbold \mu}_0(n)=\sum_{j \in {\mathcal S}_0(n)} {\mathbold
  x}_j(n)\beta_j(n)\]
with the corresponding definitions of ${\mathbold \mu}_{01}(n)$,
${\mathbold \mu}_{12}(n)$ and ${\mathbold \mu}_2(n)$

Let $ss_0(n)$ and $ss_{01}(n)$ be the sum of squared residuals after
regressing ${\mathbold Y}(n)$ on the covariates ${\mathbold x}_j(n)$
for $j$ in ${\mathcal S}_0(n)$ and ${\mathcal S}_{01}(n)$
respectively. From $k_n=o(n)$ and the Gaussian assumption on the
errors (which can be relaxed) it follows that the asymptotic values of
$ss_0(n)$ and $ss_{01}(n)$ are given by
\begin{eqnarray}
ssa_0(n)&=&{\mathbold \mu}_{12}(n)^t({\mathbold
\mu}_{12}(n)-P_0({\mathbold \mu}_{12}(n))+n\sigma^2\label{equ:asymp_1}\\
ssa_{01}(n)&=& {\mathbold \mu}_2(n)^t({\mathbold
\mu}_2(n)-P_{01}({\mathbold \mu}_2(n))+n\sigma^2\label{equ:asymp_2}
\end{eqnarray}
respectively. 
Given $n$ the  covariates are specified sequentially up to but excluding the
first covariate whose $p$-values exceeds $\alpha(n)$ for some
sequence $\alpha(n)$ tending to zero but such that
\begin{equation} \label{equ:asymp_p_alpha}
-\log(-\log(1-\alpha(n))))=o(\log(q(n))).
\end{equation} 

\begin{theorem} \label{theorem1}
Suppose that there exists a  $\tau>2$ and a $\tau'<2$ such that for all $n$
sufficiently large the following holds:\\
(i) for each proper subset ${\mathbold S}_0(n)$ of $\{1,\ldots,k(n)\}$
there exists a $\nu \in {\mathcal S}_{12}(n)$  such that
\begin{equation} \label{equ:asymp_3}
\left( 1-\frac{ssa_{01}(n)}{ssa_0(n)}\right) > \frac{\tau \log
  q(n)}{n},\\
\end{equation}
(ii) for all subsets ${\mathbold S}_0(n)$ of $\{1,\ldots,k(n)\}$
and for all $ \nu >k(n)$ 
\begin{equation} \label{equ:asymp_4}
\left( 1-\frac{ssa_{01}(n)}{ssa_0(n)}\right) <\frac{\tau' \log
  q(n)}{n}.\\
\end{equation}
Then the procedure described above is consistent.
\end{theorem}
\begin{proof} For large $n$ 
\[n\left(1-\frac{ss_{01}(n)}{ss_0(n)}\right)\asymp
n\left(1-\frac{ssa_{01}(n)}{ssa_0(n)}\right).\]
It follows from (\ref{equ:stop_rule_1}), (\ref{equ:asymp_p}), the choice
of the $\alpha(n)$ and (i) of the theorem that if
covariates ${\mathbold x}_j$ with $j\le
k(n)$ have been included then either all covariates ${\mathbold x}_j(n),\, j\le k(n)$ have
been included or there exists at least one covariate ${\mathbold x}_{\nu}(n),\,\nu
\le k(n)$ which is a candidate for inclusion. From
 (\ref{equ:stop_rule_1}), (\ref{equ:asymp_p}), the choice of
 $\alpha(n)$  and (ii) it follows that there is no covariate
 ${\mathbold x}_{\nu}(n)$ with $\nu >k(n)$ which is a 
candidate for inclusion. The procedure therefore continues until all
covariates ${\mathbold x}_j(n),\, j\le k(n)$ are included and then
terminates.
\end{proof}

It is perhaps worth emphasizing that the theorem does not require that
the covariates be uncorrelated. To take a concrete example put
$q(n)=2$ so that
\[{\mathbold Y}(n)={\mathbold x}_1(n)\beta_1(n)+{\mathbold
  x}_2(n)\beta_2(n) +{\mathbold \varepsilon}(n)\]
 where $\Vert {\mathbold x}_1(n)\Vert_2=\Vert {\mathbold
   x}_2(n)\Vert_2$ and the correlation between ${\mathbold x}_1(n)$
 and ${\mathbold x}_2(n)$ is $\rho(n)$. Suppose that $\beta_1(n)^2 >
 \beta_2(n)^2$ so that ${\mathbold x}_1(n)$ is the first candidate for
 inclusion. Then for large $n$ ${\mathbold x}_1(n)$ will be included
if
\[(\beta_1(n)+\rho\beta_2(n))^2 \ge \frac{\gamma\cdot
  \text{qchisq}(\sqrt{1-\alpha(n)},1)}{n}(\sigma^2+(1-\rho(n)^2)\beta_2(n)^2)\] 
for some $\gamma >1$.

If however the covariates are orthogonal then (\ref{equ:asymp_3})
simplifies and becomes 
\begin{equation} \label{equ:asymp_11}
\beta_{\nu}(n)^2 > \left(\sum_{j\in  {\mathcal
    S}_2(n)}\beta_j(n)^2+\sigma^2\right)\frac{\tau \log q(n)}{n}.
\end{equation}

 Theorem 1 of
\cite{LOKTAYTIB214} gives a consistency result for orthogonal
covariates with $L_2$ norm one, namely (in the present notation) 
\begin{equation} \label{equ:lasso_1}
\min_{1 \le j \le k(n)} \vert\beta_j\vert -\sigma \sqrt{2 \log
  q(n)} \rightarrow \infty.
\end{equation}
If more care is taken (\ref{equ:asymp_11}) can be expressed in the
same manner with $\tau=2$. At first glance (\ref{equ:asymp_11}) differs from
(\ref{equ:lasso_1}) by the inclusion of $\,\sum_{j\in  {\mathcal
    S}_2(n)}\beta_j(n)^2$. However lasso requires a value for
$\sigma^2$ which causes problems , particularly if $q(n)>k(n)$ (see
\cite{LOKTAYTIB214}). The term $\sum_{j\in  {\mathcal
    S}_2(n)}\beta_j(n)^2+\sigma^2$ is nothing more 
than the approximation for  $ss_0(n)/n$ and can therefore be interpreted,
if one wishes, as an estimate for $\sigma^2$ based on the sum of
squared errors of the covariates which are active at this stage.

\subsection{False discovery rate}
In this section we suppose that the data are generated as under
(\ref{equ:standard_model}) 
with ${\mathbold X}(n)=({\mathbold x}_1(n),\ldots {\mathbold
  x}_{k(n)}(n))$ given and the errors ${\mathbold \varepsilon}(n)$ are
Gaussian white noise. We suppose further that there are $q(n)$ additional
covariates ${\mathbold x}_j(n),j=k(n)+1,\ldots , k(n)+q(n)$ which are
Gaussian white noise independent of the  ${\mathbold x}_j,j=1,\ldots ,
k(n)$ and the ${\mathbold  \varepsilon}(n)$. A false discovery is the
inclusion of a variable ${\mathbold x}_j, k(n)+1\le j\le k(n)+q(n)$ in
the final active set. We denote the number of false discoveries by
$\nu$. The following theorem holds.

\begin{theorem}
Let $F_0$ be the event that none of the $k_n$ covariates ${\mathbold x}_1(n),\ldots {\mathbold
  x}_{k(n)}(n))$ are included and $F_1$ the event they are all and
moreover the first to be included. Then 
 \[\alpha \le \ex(\nu|F_i)\le\alpha/(1-\alpha), i=0,1.\]
\end{theorem}
\begin{proof}Consider $F_0$ and suppose that $j$ of the random
  covariates are active and consider the probability that these are
  indexed by $i_1,\ldots,i_j$ in the
  order of inclusion. At each stage the residuals are independent of
  the covariates not yet included due to the independence of the
  random covariates and the errors $\varepsilon$. The probability that $i_{\ell}$ is included is
  the probability that the threshold is not exceeded which is
  approximately $\alpha/q(n)$ (see the Appendix). Thus the
  probability that the included covariates are indexed by
  $i_1,\ldots,i_j$ is approximately $(\alpha/q(n))^j$. There are
  $q(n)(q(n)-1)\ldots (q(n)-j+1)$ choices for the covariates and hence
 \[\pr(\nu=j|F_0)\le q(n)^j(\alpha/q(n))^j=\alpha^j.\]
Summing over $j$ gives
\[\ex(\nu|F_0) \le \alpha/(1-\alpha).\]
The same argument works for $F_1$. \qed
\end{proof}
Sufficient conditions for the $\pr(F_i)$ to be large can be given. For
example for $F_1$ suppose the fixed covariates  ${\mathbold x}_1(n),\ldots {\mathbold
  x}_{k(n)}(n))$ satisfy
\begin{equation} \label{equ:asymp_p_F_1}
\left(1-\frac{ss_{01}}{ss_0}\right) >((2+\delta)\log
q(n)-\log\log(q(n))-2\log(-\log(1-\alpha)))/n.\nonumber
\end{equation}
for some $\delta >0$.  Then the proof of (\ref{equ:asymp_p}) given in
the appendix shows that the probability that (\ref{equ:asymp_p_F_1})
holds for one the random covariates is asymptotically $\alpha
q(n)^{-\delta/2}$. From this it follows $\pr(F_1)\ge 1-\alpha k(n)
q(n)^{-\delta/2}$ for large $n$.

Simulations suggest that the result holds under more general
conditions. As an example we use a a simulation scheme considered in
\cite{CANDetal17}.  The covariates $X_i$ are 
Gaussian but given by  an AR(1) process with coefficient 0.5. The
dependent variable is given by the logit model
\[Y=\text{rbinom}\left(1,\frac{\exp(0.08\sum_{j=2}^{21}X_j)}
  {1+\exp(0.08\sum_{j=2}^{21}X_j)}\right).\] 
It is clear that the conditions imposed above do not hold. The
covariates are not independent and $\pr(F_1)\approx 0$ rather than
approximately 1. We consider four different sample sizes
$(n,q(n))=(500,200)$, $(n,q(n))=(5000,2000)$ and the two cases with the
values of $n$ and $q(n)$ reversed. The results for least squares using
the logit model are given in Table~ \ref{tab:fdr} for $\alpha=0.01,
0.05, 0.10$. The first line for the sample sizes gives the average
number of false discoveries, the second line gives the average number
of correct discoveries. The results are based on 1000 simulations. It
is seen that the average number of false discoveries is indeed well
described by $\alpha$. The results for least squares based and the
logit model combined with Kullback-Leibler are essentially the same.
\begin{table}
\begin{center}
{\footnotesize
\begin{tabular}{llccc}
$n$&$p$&$\alpha=0.01$&$\alpha=0.05$&$\alpha=0.1$\\
\hline
500&200& 0.012&   0.051&0.103\\
&&0.594&  1.063 &1.343\\
5000&2000&0.010&0.059&0.106\\
&&6.420&7.015&7.196\\
200&500&0.010&0.043&0.102\\
&&0.029&0.130&0.179\\
2000&5000&0.006&0.058&0.112\\
&&2.92&3.477&3.701\\
\hline
\end{tabular}
\caption{False discovery rate.\label{tab:fdr}}
}
\end{center}
\end{table}

\section{Simulations and real data}
\subsection{The ProGau, ProPre1 and ProPre2 procedures}
The procedure described in this paper with $\alpha=0.01$ will
be denoted by ProGau. Two modifications, ProPre1 and ProPre2, are also 
considered both of which are intended to ameliorate the effects of high
correlation between the covariates. They are based on \cite{JIAROH15}
and are as follows. Let
\[{\mathbold x}(n)= {\mathbold u}(n){\mathbold d}(n){\mathbold v}(n)^T\]
be the singular value decomposition of the covariates ${\mathbold
  x}(n)$. Put 
\[{\mathbold f}(n)={\mathbold u}(n){\mathbold d}(n)^{-1}{\mathbold u}(n)^T\]
where ${\mathbold d}(n)^{-1}$ is the
$\min(n,q(n))$ diagonal matrix consisting of the reciprocals of
the non-zero values of ${\mathbold d}(n)$. Then ${\mathbold y}(n)$ is replaced by
${\tilde {\mathbold y}}(n)={\mathbold f}(n){\mathbold y}(n)$ and ${\mathbold
  x}(n)$ by  ${\tilde {\mathbold x}}(n)={\mathbold f}(n){\mathbold x}(n)$. The
ProPre1 procedure is to apply the ProGau procedure to $({\tilde
  {\mathbold y}}(n),{\tilde {\mathbold x}}(n))$.

The third procedure to be denoted by ProPre2 is to first
precondition the data as above but then to set $\alpha=0.5$ in
(\ref{equ:stop_rule_1}). This gives a list of candidate covariates. A
simple linear regression is now performed using these 
covariates. Covariates with a large $p$-value in the linear
regression are excluded and the remaining covariates
accepted. A large $p$-value is defined as follows: given
$\alpha=0.01$ as in ProGau and ProPre1 the cut-off $p$-value in the
linear regression to be $1-(1-\alpha)^{1/(q(n)-\nu_0)}$ as in
(\ref{equ:stop_rule_1}). As $\alpha$ is small and $q(n)$ large with 
respect to $\nu_0$ this corresponds approximately to a $p$-values of
$\alpha/q(n)$. This is somewhat ad hoc it specifies a well defined
procedure which can be compared with other procedures.

At first sight it may appear that the use of
$p$-values deriving from a linear regression implies the acceptance of
the linear model (\ref{equ:standard_model}). This is not so. In
\cite{DAV16a} non-significance regions are defined based on random
perturbations of the covariates. The resulting regions are
asymptotically the same as the confidence regions derived from the
linear model (\ref{equ:standard_model}) but do not suppose it.

\subsection{Simulations}

{\bf Linear regression}\\

The first set of simulations we report are the equi-correlation
simulations described in Section~4.1 of \cite{DEGBUERITDEZ14}. The
sample size is $n=100$ and the number of covariates is $q(n)=500$. The
covariates are generated as Gaussian random variables with covariance
matrix $\Sigma$ where $\Sigma_{ij}=0.8, i\ne j$ and
$\Sigma_{jj}=1$. Four different scenarios are considered. The number
of non-zero coefficients in the data generate according to the linear
model is either $s_0=3$ or $s_0=15$ and they are all either
i.i.d. $U(0,2)$ or i.i.d. $U(0,4)$. The error term in the model is
i.i.d $N(0,1)$. 

Although  \cite{DEGBUERITDEZ14} is concerned primarily with confidence
interval Tables~5-8 give the results of simulations for testing
all the hypotheses $H_{0,j}:  \beta_j=0$. This is equivalent to
choosing those covariates ${\mathbold x}_j$ for which $H_{0j}$ is rejected. 
The measures of performance are the power and the family-wise error
rate FWER. The power is the proportion of correctly identified active
covariates. The FWER is the proportion of times the estimated set of
active covariates contains a covariate which is not active. In Tables
6 and 8 of  \cite{DEGBUERITDEZ14} the active covariates are chosen at
random.  As ProGau, ProPre1 and ProPre2 are equivariant with respect to
permutations of the covariates the active covariates are either
$S_0=\{1,2,3\}$ or $S_0=\{1,\ldots,15\}$ as in Tables 5 and 7 of
\cite{DEGBUERITDEZ14}.

The results given in Table~\ref{tab:power_fwer}. The results for
Lasso-Pro are taken from \cite{DEGBUERITDEZ14}. As is seen from
Table~\ref{tab:power_fwer} the best overall procedure is ProPre2. In the
case $s_0=3$ and $U(0,2)$ coefficients it is slightly worse that
Lasso-Proc but the difference could well be explicable by simulation
variation. For ProGau, ProPre1 and ProPre2 500 simulations were performed.
\begin{table}
\begin{center}
{\footnotesize
\begin{tabular}{lcccccccc}
\hline
&\multicolumn{4}{c}{$s_0=3$}&\multicolumn{4}{c}{$s_0=15$}\\
&\multicolumn{2}{c}{$U(0,2)$}&\multicolumn{2}{c}{$U(0.4)$}&\multicolumn{2}{c}{$U(0,2)$}&\multicolumn{2}{c}{$U(0,4)$}\\
Lasso-Pro&0.56&0.10&0.79&0.11&0.70&1.00&0.92&1.00\\
ProGau&0.60&0.19&0.79&0.07&0.25&0.97&0.46&0.96\\
ProPre1&0.41&0.01&0.71&0.01&0.09&0.01&0.29&0.00\\
ProPre2&0.55&0.12&0.79&0.07&0.44&0.26&0.73&0.09\\
\hline
\end{tabular}
\caption{Power and family-wise error rate for the lasso-projection
  method of \cite{DEGBUERITDEZ14} and the procedures ProGau, ProPre1
  and ProPre2 as described in this section. \label{tab:power_fwer}}
}
\end{center}
\end{table}

Tables~1-4 of \cite{DEGBUERITDEZ14} give the results of some
simulations on the covering frequencies and the lengths of
confidence intervals. Again we shall restrict the comparisons to the
four cases detailed above.
The average cover for the active set $S_0$ is defined by
\[\text{Avgcov}=\frac{1}{s_0}\sum_{j \in S_0}\pr(\beta_j^0 \in CI_j)\]
and the average length  by
\[\text{Avglength}=\frac{1}{s_0}\sum_{j \in S_0}\text{length}(CI_j)\]
with analogous definitions for the complement $S_0^c$. The $CI_j$ are
the confidence intervals for the coefficients $\beta_j^0$ used to
generate the data.

The $0.95$-confidence intervals for the covariates for the procedures
ProGau, ProPre1 and ProPre2 are calculated as follows. Let $\tilde{S}_0$
denote the estimated active set. For the covariates in $\tilde{S}_0$
the confidence intervals are those calculated from a simple linear
regression. For a covariate in $\tilde{S}_o^c$ the confidence interval
is calculates by appending this covariate to the set $\tilde{S}_0$,
performing a linear regression and using confidence interval from this
regression. It is pointed out again that these confidence
intervals are not dependent on the linear model
(\ref{equ:standard_model}).

The results are given in Table~\ref{tab:cov_len} with the Lasso-Pro
results taken from  \cite{DEGBUERITDEZ14}. For $s_0=3$ and $s_0=15$ the
standard deviations of the estimates of the $\beta_j$ are approximately
0.186 and 0.216 respectively. Thus for a coverage probability of
$\gamma$ the optimal lengths of the confidence intervals are
approximately $2\cdot 0.186\cdot\text{qnorm}((1+\gamma)/2)$ and
$2\cdot 0.216\cdot \text{qnorm}((1+\gamma)/2)$ respectively. For the Lasso-Pro
procedure with the stated coverage probabilities the lengths of the
confidence intervals for $S_0$ with $s_0=3$ are 78\% for $U(0,2)$ and
74\%  for $U(0,4)$. For $S_0^c$ and $s_0=3$ the corresponding
percentages are 10\% for $U(0,2)$ and 5\% for $U(0,4)$. For $S_0$ with
$s_0=15$ the percentages are 73\% for $U(0,2)$ and 85\% for
$U(0,4)$. For $S_0^c$ the intervals are approximately 25\% shorter
than the optimal intervals. The explanation would seem to be the
following (personal communication from Peter B\"uhlmann).  The
theorems of \cite{DEGBUERITDEZ14} require $s_0=o(\sqrt{n}/\log q(n))$  to
guarantee an asymptotically valid uniform approximation and prevent
super efficiency. On plugging in $n=100$ and $q(n)=500$ gives
$s_0=o(1.61)$ so that $s_0=15$ is, so to speak,  not in the range of
applicability of the theorems. The lengths for $S_0^c$ with $s_0=3$ are
only slightly longer than the optimal intervals so even here there may
be a super efficiency effect. 

We note that in this simulation ProPre1 is dominated by ProGau.
\begin{table}
\begin{center}
{\footnotesize
\begin{tabular}{llcccccccc}
\hline
&&\multicolumn{4}{c}{$s_0=3$}&\multicolumn{4}{c}{$s_0=15$}\\
&&\multicolumn{2}{c}{$U(0,2)$}&\multicolumn{2}{c}{$U(0.4)$}&\multicolumn{2}{c}{$U(0,2)$}&\multicolumn{2}{c}{$U(0,4)$}\\
Lasso-Pro&$S_0$&0.89&0.82&0.87&0.82&0.56&0.56&0.53&0.55\\
&$S_0^c$&0.95&0.80&0.96&0.80&0.93&0.57&0.93&0.56\\
ProGau&$S_0$&0.84&0.70&0.89&0.73&0.81&1.42&0.84&1.81\\
&$S_0^c$&0.91&0.76&0.93&0.78&0.78&1.43&0.84&1.82\\
ProPre1&$S_0$&0.52&0.55&0.79&0.72&0.03&1.65&0.24&2.86\\
&$S_0^c$&0.48&0.63&0.80&0.79&0.02&1.71&0.23&2.96\\
ProPre2&$S_0$&0.76&0.68&0.87&0.73&0.62&1.24&0.85&1.08\\
&$S_0^c$&0.77&0.75&0.89&0.78&0.60&1.26&0.87&1.09\\
\hline
\end{tabular}
\caption{Average covering frequency and average length of interval for
  the four procedures, for $S_0$ and $S_0^c$ and for the four
  simulation scenarios with specified covering probability 0.95.  The
  Lasso-Pro results are taken from  \cite{DEGBUERITDEZ14}. \label{tab:cov_len}} 
}
\end{center}
\end{table}
Simulation experiments are reported in  \cite{JIAROH15}. The ones to be
given here correspond to those of Figure~4 of \cite{JIAROH15}. The
sample size is $n=250$ with $q(n)=250, 5000, 10000, 15000$ and
30000. The covariates have covariance matrix $\Sigma$ with
$\Sigma_{ii}=1$ and $\Sigma_{ij}=\rho, i\ne j$ with $\rho=0.85$ and
$\rho=0.05$. The number of active covariates is $k=20$ with
coefficients  $\beta_1=\ldots=\beta_{20}=3$. The noise is
i.i.d. $N(0,1)$. The results are given in Table~\ref{tab:fn_fp}.

It is seen that ProPre2 gives the best results for $q(n)=5000,10000,
15000$ and 30000 but fails in terms of false negatives for
$q(n)=250$. More generally ProPre2 fails for $q(n)$ in the range
230-280. This seems to correspond to the peaks in the black `puffer'
lines of Figure~4 of \cite{JIAROH15}. An explanation of this
phenomenon is given in  \cite{JIAROH15}. For  $q(n)$ in the range
280-30000 ProPre2 outperforms ProGau and ProPre2 and also all the methods
in Figure~4 of \cite{JIAROH15}, particularly with respect to the
number of false positives.  
\begin{table}
\begin{center}
{\footnotesize
\begin{tabular}{lcccccccccc}
\hline
&\multicolumn{10}{c}{$\rho=0.85$}\\
&\multicolumn{2}{c}{250}&\multicolumn{2}{c}{5000}&\multicolumn{2}{c}{10000}&\multicolumn{2}{c}{15000}&\multicolumn{2}{c}{30000}\\
ProGau&0.00&1.55&9.12&7.87&15.1&9.32&17.3&10.0&18.1&9.91\\
ProPre1&20.0&0.00&15.7&0.01&17.7&0.00&17.2&0.00&18.6&0.01\\
ProPre2&19.3&0.44&0.00&0.02&0.00&0.07&0.47&0.16&0.85&1.08\\
&\multicolumn{10}{c}{$\rho=0.10$}\\
&\multicolumn{2}{c}{250}&\multicolumn{2}{c}{5000}&\multicolumn{2}{c}{10000}&\multicolumn{2}{c}{15000}&\multicolumn{2}{c}{30000}\\
ProGau&0.00&0.18&3.38&2.96&8.04&4.11&11.7&5.00&16.1&5.94\\
ProPre1&14.3&0.40&13.48&0.00&15.8&0.00&17.4&0.00&18.2&0.01\\
ProPre2&17.0&0.00&0.00&0.02&0.00&0.04&0.11&0.06&0.81&0.09\\
\hline

\end{tabular}
\caption{Average number of false negatives (first number) and average
  number of false positives (second number) for the procedures ProGau,
  ProPre1 and ProPre2 for different numbers of covariates and two values of
  $\rho$, 0.85 and 0.1. \label{tab:fn_fp}}  
}
\end{center}
\end{table}

{\bf Logistic regression}\\
Logistic regression was also considered in
\cite{DEGBUERITDEZ14}. Table~\ref{tab:log_kl} gives the results of some
simulations as described in \cite{DEGBUERITDEZ14}: it includes Table~9
of \cite{DEGBUERITDEZ14}. The sample size is $n=100$ and the number of
covariates is $q(n)=500$. The covariates ${\mathbold x}(n)$  are
generated as Gaussian random variables with mean zero and covariance
matrix $\Sigma$ with $\Sigma_{ij}=0.8^{\vert i-j\vert}$ for $i \ne j$ and
$\Sigma_{ii}=1$. The parameter ${\mathbold \beta}$ is given by
$\beta_i=U(0,\eta), i=1,2,3$ with $\eta=1,2,4$ and $\beta_i=0, i\ge
4$. The dependent variables $y_i$ are generated as a independent
binomial random variables $\text{binom}(p_i,1)$  with $p_i=\exp({\mathbold x}^t{\mathbold
  \beta}_i)/(1+\exp({\mathbold x}^t{\mathbold \beta}_i)$ (see
\cite{DEGBUERITDEZ14}). The KL-procedure of the table is as described
in Section~2.4 . The cut-off $p$-value is $0.01$ so that all covariates
are included up to but excluding the first with a $p$-value exceeding 0.01.

\begin{table}[h]
\begin{center}
\begin{tabular}{lcccc}
\multicolumn{5}{c}{Logistic regression}\\
\hline
\multicolumn{2}{c}{ }&\multicolumn{3}{c}{Toeplitz}\\
\cline{3-5}
{\bf measure}&{\bf method}&$U(0,1)$&$U(0,2)$&$U(0,4)$\\ 
\hline
Power&Lasso-ProG&0.06&0.27&0.50\\
&MS-Split&0.07&0.37&0.08\\
&KL-procedure&0.26&0.32&0.37\\
FWER&Lasso-ProG&0.03&0.08&0.23\\
&MS-Split&0.01&0.00&0.00\\
&KL-procedure&0.03&0.02&0.01\\
\hline
\end{tabular}
\caption{This incorporates Table~9 of \cite{DEGBUERITDEZ14} and the
procedure described in the text based on the Kullback-Leibler
discrepancy\label{tab:log_kl}}
\end{center}
\end{table}
\subsection{Real data}
The real data includes three of the data sets used in
\cite{DETBUH03}, namely Leukemia (\cite{GOLETAL99},
http://www-genome.wi.mit.edu/cancer/.), Colon (\cite{ALOETAL99},
http://microarray.princeton. edu/oncology/.)  and
Lymphoma (\cite{ALIETAL00}, http://llmpp.nih.gov/
lymphoma/data/figure1). A fourth data set SRBCT is included. All data
sets were downloaded from\\
\url{http://stat.ethz.ch/~dettling/bagboost.html}\\
\quad\\
A fifth data set, prostate cancer {\it prostate.rda} was downloaded
from the {\it   lasso2} package available from the CRAN R package
repository.

The dependent variable ${\mathbold y}(n)$ is integer valued in all
cases, each integer denoting a particular form, or absence, of
cancer.\\
\quad\\
{\bf Least squares}\\
The data set Leukemia was analysed in Section~\ref{sec:2.1} using
ProGau with the $p$-values of the first five genes given
(\ref{equ:leukemia_lsq_proc1}). Classification of the form of cancer
based on these three genes results in one misclassification. 
If this is repeated using ProPre1 the result
is
\begin{equation} \label{equ:leukemia_lsq_proc2}
\begin{tabular}{cccccc}
gene number&979&2727&1356&2049&3054\\
$p$-value&0.0072&0.1324&0.7365&0.3646&0.7324\\
\end{tabular}
\end{equation}
Classifying the cancer using gene 979 only results in 15
misclassifications. The procedure ProPre2 also results in the single
gene 979 and 15 misclassifications. The results for the remaining
real data sets given above are equally poor so the procedures ProPre1
and ProPre2 will not be considered any further.



Table~\ref{tab:gene_expres} gives the results for ProGau for all five
data sets. The second row gives the sample size, the number of
covariates followed by the number of different cancers.
The rows 4-8 give the covariates in order and their $p$-values. These
are included in order up to but excluding the first covariate with a
$p$-value exceeding 0.01. The row 9 gives the number of
misclassifications based on the included covariates. Row 10
gives the number of possibly relevant covariates calculates as
described in  Section~\ref{sec:2.1}.

In previous versions of this paper the number of misclassifications 
based on ProGau as described above was compared with the results given
in Table~1 of \cite{DETBUH03}. The comparison is illegitimate. The
results in \cite{DETBUH03} are calculated by cross validation: each
observation is in turn classified using the remaining $n-1$
observations and Table~1 of \cite{DETBUH03} reports the number of
misclassification using different procedures. Line 11 gives the
results for the following procedure based on ProGau. 

The $i$th observation is eliminated leaving $n-1$
observations. Each of these is eliminated in turn and the ProGau
procedure applied to the remaining $n-2$ observations. This gives $n-1$
sets of active covariates. The ten most frequent covariates are
then used to classify the $i$th observation. This is done for $i,1\le
i\le n$. The sample is augmented by all misclassified
observations and the procedure applied to the new sample. This is
a form of boosting.  This is done 60 times or until all observations
are  correctly classified which ever happens first. The numbers in row 11
give the number of misclassifications followed by the number of
additional observations. 

\begin{table}[h]
\begin{center}
{\footnotesize
\begin{tabular}{cccccccccc}
\multicolumn{10}{c}{Linear regression: least squares}\\
\multicolumn{10}{c}{ }\\
\hline
\multicolumn{2}{c}{Leukemia}&\multicolumn{2}{c}{Colon}&\multicolumn{2}{c}{Lymphoma}&\multicolumn{2}{c}{SRBCT}&\multicolumn{2}{c}{Prostate}\\
\multicolumn{2}{c}{72, 3571, 2}&\multicolumn{2}{c}{62,
  2000, 2}&\multicolumn{2}{c}{62, 4026, 3}&\multicolumn{2}{c}{63,
  2308, 4}&\multicolumn{2}{c}{102, 6033, 2}\\
cov.&$P$-val.&cov.&$P$-val.&cov.&$P$-val.&cov.&$P$-val.&cov.&$P$-val.\\
\hline
1182& 0.0000&493 &7.40e-8& 2805& 0.0000&1389&2.32e-9&2619&0.0000\\
1219& 8.57e-4&175 &4.31e-1& 3727& 8.50e-9&1932&2.02e-7&203&5.75e-1\\
2888& 3.58e-3&1909& 5.16e-1& 632& 1.17e-4&1884&9.76e-5&1735&4.96e-1\\
1946& 2.54e-1&582 &3.33e-1& 714& 2.30e-2&1020&1.70e-1&5016&9.81e-1\\
2102& 1.48e-1&1772& 9.995e-1&2036& 2.94e-1& 246&6.75e-2&2940&9.52e-1\\
\hline
\multicolumn{2}{c}{1}&\multicolumn{2}{c}{9}&\multicolumn{2}{c}{1}&\multicolumn{2}{c}{13}&\multicolumn{2}{c}{8}\\
\multicolumn{2}{c}{281}&\multicolumn{2}{c}{45}&\multicolumn{2}{c}{1289}&\multicolumn{2}{c}{115}&\multicolumn{2}{c}{185}\\
\hline
\hline
\multicolumn{2}{c}{0, 10}&\multicolumn{2}{c}{0,
  59}&\multicolumn{2}{c}{0, 7}&\multicolumn{2}{c}{0,
  75}&\multicolumn{2}{c}{1,42}\\
\hline
\end{tabular}
\caption{Gene expression data: sample size, number of genes and number
  of cancers in the
  second row. The rows 4-8 give the first five covariates in order of
  inclusion and their $p$-values. Row 9 gives the number of
  misclassifications based on the active covariates. Row 10 gives
  the number of possibly relevant covariates. Row 11 gives the number
  of misclassifications using cross validation and the number of
  repeated observations. \label{tab:gene_expres}}
}
\end{center}
\end{table}
\quad\\
{\bf Robust regression}\\
Table~\ref{tab:rob_gene_expres} corresponds to
Table~\ref{tab:gene_expres}. In all cases the first gene is the
same. For the lymphoma and SRBCT data sets the first five genes are
the same and in the same order. Nevertheless there are differences
which may be of importance.

One simple but, because of the nature of dependent variable,
somewhat artificial example which demonstrates this it to change the
first ${\mathbold y}(n)$ value for the colon data from 0 to -10. The
first gene using least squares is now 1826 with a $p$-value of 0.352. This
results not surprisingly in 23 misclassifications of the remaining 61
tissues. For the robust regression the first gene is again 493 with a
$p$-value of 1.40e-5 which results in nine misclassifications as
before.

In the real data examples considered here the dependent variable
denotes a form of cancer and is therefore unlikely to contain
outliers. Nevertheless the stepwise regression procedure can reveal
outliers or exotic observations. Given residuals $r_i, i=1,\ldots,n$
a measure for the outlyingness of $r_i$ is $\vert
r_i\vert/s_n$ where $s_n=\text{median}(\vert r_1\vert,\ldots,\vert
r_n\vert)$. Hampel's 5.2 rule (\cite{HAM85}) is
 to identify all observations whose outlyingness value exceeds 5.2 as
 outliers. For the leukemia data based on the first three
 covariates and using least squares the observations numbered 21, 32
 and 35 were identified as  outliers with outlyingness values 5.53,
 9.78 and 5.96 respectively. The robust regression for the same data
 resulted in four outliers, 32, 33, 35 and 38 with values 18.0, 6.46,
 18.42 and 7.50 respectively. The observations 32 and 35 are
 also identified by least squares but the outlyingness values for  the
 robust regression are much larger .
\begin{table}[ht]
\begin{center}
{\footnotesize
\begin{tabular}{cccccccccc}
\multicolumn{10}{c}{Robust linear regression}\\
\multicolumn{10}{c}{ }\\
\hline\\
\multicolumn{2}{c}{Leukemia}&\multicolumn{2}{c}{Colon}&\multicolumn{2}{c}{Lymphoma}&\multicolumn{2}{c}{SRBCT}&\multicolumn{2}{c}{Prostate}\\
\multicolumn{2}{c}{72, 3571}&\multicolumn{2}{c}{62,
  2000}&\multicolumn{2}{c}{62, 4026}&\multicolumn{2}{c}{63,
  2308}&\multicolumn{2}{c}{102, 6033}\\
cov.&$P$-val.&cov.&$P$-val.&cov.&$P$-val.&cov.&$P$-val.&cov.&$P$-val.\\
\hline\\
1182& 1.11e-9&493 &2.76e-5& 2805& 1.11e-10&1389&1.65e-5&2619& 1.34e-12\\
2888& 7.41e-5&449&5.04e-1&3727& 8.82e-6&1932& 4.65e-6&1839& 1.47e-1\\
1758& 2.35e-5&1935& 1.26e-1&632&1.39e-3&1884& 7.15e-4&2260& 7.70e-3\\
3539& 4.89e-2&175&8.95e-1&714&4.04e-2&1020& 9,85e-2&1903& 8.23e-1\\
3313& 2.86e-2&792&1.86e-1&2036& 3.71e-1& 246&3.99e-2&5903&7.65e-1\\
\hline\\
\multicolumn{2}{c}{1, 1.39\% (2)}&\multicolumn{2}{c}{9, 14.52\%
  (2)}&\multicolumn{2}{c}{1, 1.61\% (3)}&\multicolumn{2}{c}{13, 20.63\% (4)}&\multicolumn{2}{c}{8, 7.77\% (2)}\\
\hline\\
\end{tabular}
\caption{As for Table~\ref{tab:gene_expres} but for robust
  regression\label{tab:rob_gene_expres}} 
}
\end{center}
\end{table}
\begin{table}[h]
\begin{center}
{\footnotesize
\begin{tabular}{cccccc}
\multicolumn{6}{c}{Logistic regression: Kullback-Leibler}\\
\multicolumn{6}{c}{ }\\
\hline\\
\multicolumn{2}{c}{Leukemia}&\multicolumn{2}{c}{Colon}&\multicolumn{2}{c}{Prostate}\\
\multicolumn{2}{c}{72, 3571}&\multicolumn{2}{c}{62,
  2000}&\multicolumn{2}{c}{102, 6033}\\
cov.&$P$-val.&cov.&$P$-val.&cov.&$P$-val.\\
\hline\\
956& 0.0000&377&2.12e-7&2619& 0.0000\\
1356& 8.14e-2&356&5.35e-2&4180& 6.69e-1\\
264& 1.00e-0&695&1.0000&1949& 9.88e-1\\
\hline\\
\multicolumn{2}{c}{6, 8.33\%}&\multicolumn{2}{c}{10, 16.13\%
}&\multicolumn{2}{c}{9, 8.82\% }\\
\hline\\
\end{tabular}
\caption{Gene expression data corresponding to
  Table~\ref{tab:gene_expres} for for logistic regression based on the
  Kullback-Leibler discrepancy. \label{tab:gene_expres_log_kl}}
}
\end{center}
\end{table}

{\bf Logistic regression}\\
For the leukemia, colon and prostate data the dependent variate
${\mathbold y}$ takes on only the values zero and one. These data
sets can therefore can be analysed using logistic regression and the
Kullback-Leibler discrepancy (see Section~\ref{sec:kul_leib}).
Table~\ref{tab:gene_expres_log_kl}  gives the results corresponding to
those of Table~\ref{tab:gene_expres}.  

The second covariates for the leukemia and colon data exceed the
cut-off $p$-value of 0.05 but could nevertheless be relevant. If they
are included the number of misclassifications are zero and four
respectively.\\

{\bf The birthday data}\\
The data consist of the number of births on every day from 1st January
1969 to the 31st December 1988. The sample size is $n=7305$. The data
are available as `Birthdays' from the R-package `mosaicData'. They
have been analysed in \cite{GELetal13}.

In a first step a trend was calculated using a polynomial of order
7 by means of a robust regression using Huber's
$\psi$-function with tuning constant $cnt=1$ (page 174 of
\cite{HUBRON09}). 

The trend was subtracted and the residuals were analysed by means of a
robust stepwise regression again using Huber's $\psi$-function with
tuning constant $cnt=1$. The covariates were the trigonometric
functions
\[xs_{j}(i)=\sin(\pi j i/n) \text{ and } xc_j(i)=\cos(\pi ji/n), \,
i,j=1,...,7305,\]
but  with the difference that the $xs$ and $xc$ were treated in pairs
$xs_j$ and $xc_j$.  The cut-off $p$-value was $p=0.01$. This resulted
in 54 pairs being included in the regression. The first five periods
in order of importance were 7 days, 3.5 days, one year, six months and
2.33 days ($1/3$ of a week). Figure~\ref{fig:birth_day_res} shows the
average size of the residual for each day of the year. The marked days
are 1st-3rd January, 14th February, 29th February, 1st April, Memorial
Day, Independence Day, 8th August (no explanation), Labor Day,
Veterans Day, Thanksgiving, Halloween, Christmas, and the 13th of each month.\\
\begin{figure}[hb]
\begin{center}
\includegraphics[width=10cm,height=15cm,angle=270]{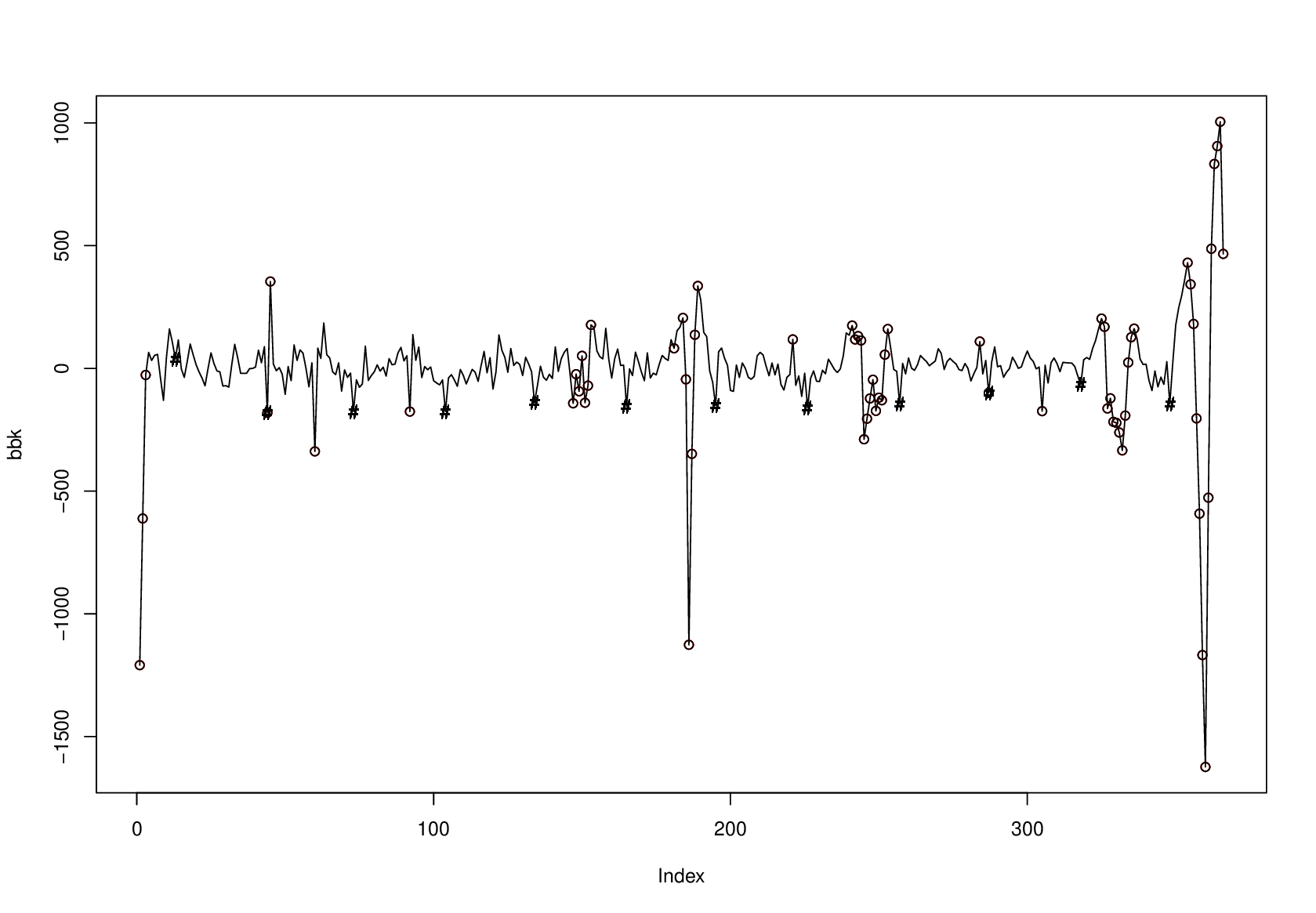}
\end{center}
\caption{The average size of the residuals for each day of the year
  for the birthday data. \label{fig:birth_day_res}}  
\end{figure}

{\bf Graphs}\\
Given variables ${\mathbold x}(n)$ a graphical representation of their
mutual dependencies can be constructed. The nodes of the graph
represent the variables and the edges indicate the existence of
dependencies between the joined nodes. A graph can be construct by
performing a stepwise regression of each variable on the others and
then joining that variable with those variables in the active subset
for that variable. In all there are $q(n)$ regressions and to take
this into account the cut-off $p$-value $\alpha$ is replaced by
$\alpha/q(n)$. For the colon data with $\alpha=0.05$ this results in a
graph with 1851 edges and a computing time of 45 seconds. For the
prostate data the graph has 7178 edges and the computing time was 15.5
minutes
.

In \cite{MEIBUE06} the authors give a lasso based method for choosing
a graph. A major problem with lasso is the choice of the regularizing
parameter $\lambda$ .  In the case of graphs however the authors show
that there is a criterion which allows a specification of
$\lambda$: if $\lambda$ is given as in (9) on page 1446
of their paper then $\alpha$ is an upper bound for the probability of
falsely joining two distinct connected components of the graph.

As an example they simulated a Gaussian sample with $n=600$ and $k=1000$ based
on a random graph with 1747 edges.  Unfortunately there is an error in
the description of the construction of the graph. The stated
probability $\varphi(d/\sqrt{p})$ (bottom line of page 1447) of joining two given
nodes is not consistent with the graph in Figure~1 of the paper (confirmed in a
personal communication by Nicolai Meinshausen). In the reconstruction
of the graph 1109 of the 1747 edges were correctly identified and there
were two false positives.  An attempt was made to construct a graph
similar to that of \cite{MEIBUE06}. The resulting graph had 1829 edges. In 10
simulations the average number of correctly identified edges was 1697
and on average there were 0.8 false positives. These results are
better than those based on the lasso. This is one more indication that
the stepwise procedure of \cite{DAV16c} is not only much simpler than
lasso it also gives better results.

\section{Acknowledgment}
The author thanks Lutz D\"umbgen for the exact $p$-value
(\ref{equ:pval_1_exact}) and the proof. He also thanks Peter
B\"uhlmann for help with \cite{DEGBUERITDEZ14}, Nicolai Meinshausen
for help with  \cite{MEIBUE06}, Joe Whittaker for the section on
graphs, Christian Hennig and Oliver Maclaren for helpful comments on
earlier versions.

\section{Appendix}
{\bf Proof of (\ref{equ:beta_p})}.\\
The result and the following proof are due to Lutz D\"umbgen. Given ${\mathbold y}(n)$
and $k$ linearly independent covariates ${\mathbold
x}_1(n),\ldots,{\mathbold x}_k(n)$ with $k<n-1$ consider the best linear
approximation of  ${\mathbold y}(n)$ by the ${\mathbold x}_j(n)$ in the
$\Vert\cdot\Vert_2$ norm:
\[ (\beta_1^*,\ldots,\beta_k^*)=\argmin_{\beta_1,\ldots,\beta_k}\Vert {\mathbold
  y}(n)-\sum_{j=1}^k{\mathbold x}_j(n)\beta_j\Vert_2.\]
The residual vector ${\mathbold r}_k(n)={\mathbold
  y}(n)-\sum_{j=1}^k{\mathbold x}_j(n)\beta_j^*$ lies in the orthogonal
complement $V_0$ of the space $X_0$ spanned by the ${\mathbold
  x}_1,\ldots,{\mathbold x}_k$. The space $V_0$ is of dimension
$n-k>1$. Let ${\mathbold Z}(n)\notin X_0$ be an additional vector and
consider the best linear approximation to ${\mathbold y}(n)$ based on
${\mathbold x}_1,\ldots,{\mathbold x}_n,{\mathbold Z}(n)$. This is
equivalent to considering the  best linear approximation to ${\mathbold y}(n)$ based on
${\mathbold x}_1,\ldots,{\mathbold x}_n,Q_0({\mathbold Z}(n))$ where
$Q_0$ is the orthogonal projection onto $V_0$. The sum of squared
residuals of this latter 
\[SS_1=ss_0-\frac{({\mathbold r}(n)^tQ_0{\mathbold Z}(n))^2}{\Vert
  Q_0{\mathbold Z}(n)\Vert_2^2}\]
where $ss_0=\Vert {\mathbold r}(n)\Vert_2^2$. If ${\mathbold Z}(n)$ is
standard Gaussian white noise then $Q_0{\mathbold Z}(n)$ is standard
Gaussian white noise in $V_0$. Let ${\mathbold v}_1(n),\ldots,{\mathbold
  v}_{n-\nu_0}(n)$ be an orthonormal basis of $V_0$ such that ${\mathbold
    r}(n)=\sqrt{ss_0}{\mathbold v}_1(n)$. Then
\[SS_1=ss_0\left( 1-\frac{({\mathbold v}_1(n)^tQ_0{\mathbold
      Z}(n))^2}{\sum_{j=1}^{n-\nu_0}({\mathbold v}_j(n)^tQ_0{\mathbold
      Z}(n))^2}\right)\]
which has the same distribution as
\[ss_0\left( 1-\frac{Z_1^2}{\sum_{j=1}^{n-\nu_0}Z_j^2}\right)\]
where the $Z_j$ are i.i.d $N(0,1)$. This proves (\ref{equ:beta_p}) on
noting that the $Z_j^2$ are independent $\chi_1^2$ random variables so
that
\[\frac{Z_1^2}{\sum_{j=1}^{n-\nu_0}Z_j^2}\stackrel{D}{=}\frac{\chi_1^2}{\chi_1^2+
  \chi_{n-\nu_0-1}^2}\stackrel{D}{=}B_{1/2,(n-\nu_0-1)/2}.\]
\quad\\
{\bf Proof of (\ref{equ:asymp_p})}.\\
Let the required $p$-value for stopping be $\alpha$. Then solving
(\ref{equ:pval_1_exact}) for small $\nu_0$ leads to 
\[\alpha=1-\text{pbeta}(x,1/2,(n-1)/2)^{q(n)}.\]
Now
\begin{eqnarray*}
\text{pbeta}(x,1/2,(n-1)/2)&=&1-\frac{\Gamma(n/2)}{\Gamma(1/2)\Gamma((n-1)/2)}\int_x^1u^{-1/2}(1-u)^{(n-1)/2}\,du\\
&\asymp&1-\frac{1}{\sqrt{2\pi}}\int_x^{\infty}u^{-1/2}\exp(-nu/2)\,du\\
&\asymp&1-\sqrt{\frac{2}{\pi}}\frac{1}{\sqrt{nx}}\exp(-nx/2)\\
&=&1-\exp(-0.5nx-0.5(\log nx)+0.5\log(2/\pi)).
\end{eqnarray*}
Thus
\begin{eqnarray*}
\text{pbeta}(x,1/2,(n-1)/2)^{q(n)}&\asymp&\left(1-\exp(-0.5nx-0.5(\log nx)+0.5\log(2/\pi))\right)^{q(n)}\\
&\asymp&\exp(-q(n)\exp(-0.5nx-0.5(\log nx)+0.5\log(2/\pi)))
\end{eqnarray*}
which leads to
\[\log(-\log(1-\alpha))\approx\log q(n) -0.5nx-0.5\log nx+0.5\log(2/\pi).\]
On putting
\[x=(2\log q(n)-\log\log(q(n))-2\log(-\log(1-\alpha))-\log\pi)/n\]
it is seen that
\[\log q(n) -0.5nx-0.5\log nx+0.5\log(2/\pi)=\log(-\log(1-\alpha))+o(1).\]

\bibliographystyle{apalike}
\bibliography{literature}
\end{document}